\documentclass[letterpaper, 10 pt, conference]{ieeeconf}  % Comment this line out if you need a4paper

\IEEEoverridecommandlockouts                              % This command is only needed if 
\makeatletter

\let\proof\@undefined
\let\endproof\@undefined
\makeatother

\title{\LARGE \bf
Safety-Critical Control of Stochastic Systems using \\Stochastic Control Barrier Functions
}

\author{Chuanzheng Wang$^\dagger$, Yiming Meng$^\dagger$, Stephen L. Smith, Jun Liu% <-this % stops a space
\thanks{$^\dagger$Equal contribution}
\thanks{Chuanzheng Wang, Yiming Meng and Jun Liu are with the Department of Applied Mathematics,        
University of Waterloo, Waterloo, Ontario, Canada, 
        {\tt\small \{cz.wang, yiming.meng, j.liu\}@uwaterloo.ca}}%
\thanks{Stephen L. Smith is with the Department of Electrical and Computer Engineering, 
University of Waterloo, Waterloo, Ontario, Canada, 
        {\tt\small stephen.smith@uwaterloo.ca}}%
}

\usepackage{amsmath,amssymb,amsthm,amstext}
\usepackage{cite}
 
 \usepackage{algorithm,algorithmic}
 \usepackage{graphicx}
\usepackage{subcaption}
 \usepackage{mwe}
 \renewcommand{\epsilon}{\varepsilon}
 \renewcommand{\theta}{\vartheta}
 \renewcommand{\rho}{\varrho} % remember my teacher and friend Adalberto!
 \renewcommand{\phi}{\varphi}

 \newcommand{\ub}{\mathbf{u}}

 \newtheorem{definition}{Definition}[section]
 \newtheorem{remark}[definition]{Remark}
 \newtheorem{prop}[definition]{Proposition}
 \newtheorem{thm}[definition]{Theorem}

 \newtheorem{ass}[definition]{Assumption}
 \newtheorem{prob}[definition]{Problem}
\usepackage{xcolor}

\DeclareMathOperator*{\argmin}{arg\,min}

\usepackage{booktabs}
\usepackage{mathrsfs}
\usepackage{dsfont}

\newcommand{\R}{\mathbb{R}}
\newcommand{\pp}{\mathbb{P}}
\newcommand{\U}{\mathcal{U}}
\newcommand{\X}{\mathcal{X}}

\newcommand{\E}{\mathbb{E}}
\newcommand{\Aa}{\mathcal{A}}
\DeclareMathOperator{\Tr}{tr}

\begin{document}

\maketitle
\thispagestyle{empty}
\pagestyle{empty}

%%%%%%%%%%%%%%%%%%%%%%%%%%%%%%%%%%%%%%%%%%%%%%%%%%%%%%%%%%%%%%%%%%%%%%%%%%%%%%%%
\begin{abstract}
Control barrier functions have been widely used for synthesizing safety-critical controls, often via solving quadratic programs. However, the existence of Gaussian-type noise may lead to unsafe actions and result in severe consequences. In this paper, we study systems modeled by stochastic differential equations (SDEs) driven by Brownian motions. We propose a notion of stochastic control barrier functions (SCBFs) and show that SCBFs can significantly reduce the control efforts, especially in the presence of noise, compared to stochastic reciprocal control barrier functions (SRCBFs), and offer a less conservative estimation of safety probability, compared to stochastic zeroing control barrier functions (SZCBFs). Based on this less conservative probabilistic estimation for the proposed notion of SCBFs, we further extend the results to handle high relative degree safety constraints using high-order SCBFs. We demonstrate that the proposed SCBFs achieve good trade-offs of performance and control efforts, both through theoretical analysis and numerical simulations. 

%(SRCBFs), as the latter guarantee a pathwise almost surely safety at the cost of possible impulse-like control inputs. In addition, SCBFs quantify the worst-case pathwise safety probability with higher quality than the  verification results  by stochastic zeroing control barrier functions (SZCBFs). For this reason, SCBFs are more applicable for handling high relative degree safety constraints and we further extend the quantification results to such situations based on SCBFs.
\end{abstract}

%%%%%%%%%%%%%%%%%%%%%%%%%%%%%%%%%%%%%%%%%%%%%%%%%%%%%%%%%%%%%%%%%%%%%%%%%%%%%%%%
\section{INTRODUCTION}\label{sec:intro}

%\subsection{Background and Literature Review}

For some real-world control problems, safety-critical control must be used in order to prevent severe consequences. It requires not only achieving control objectives, but also providing control actions with guaranteed safety \cite{garcia2015comprehensive}. Hence incorporating safety criteria is of great importance in practice when designing controllers. These requirements need to be satisfied in practice both with or without noise and disturbance. In \cite{lamport1977proving}, the notion of safety control was first proposed in the form of correctness. It was then formalized in \cite{alpern1985defining}, where the authors stated that a safety property stipulates that some ``bad thing'' does not happen during execution. There has been extensive research in safety verification problems using , e.g. discrete approximations \cite{ratschan2007safety} and computation of reachable sets \cite{girard2006efficient}.

One recent framework is to use control barrier functions (CBFs) to deal with safety criteria \cite{ames2019control}. CBFs are combined with control Lyapunov functions (CLFs) as constraints of quadratic programming (QP) problems in \cite{ames2014control}. The authors show that safety criteria can be transformed into linear constraints of the QP problems for control inputs. By solving the QP problems, we will find out a sequence of actions that generate safe trajectories during execution. It is shown in \cite{rauscher2016constrained} that finding safe control inputs by solving QP problems can be extended to an arbitrary number of constraints and any nominal control law. As a result, CBFs are widely used recently in a variety of applications such as lane keeping \cite{ames2016control} and obstacle avoidance \cite{chen2017obstacle}. Since solving QP problems using CBFs requires that the derivative of CBFs is dependent of the control input, which is not always the case for real time application of robotics \cite{hsu2015control}, CBFs are extended to handle high order relative degree as in \cite{nguyen2016exponential} and \cite{xiao2019control}. The authors of \cite{nguyen2016exponential} propose a way of designing exponential control barrier functions (ECBFs) using input-output linearization and in \cite{xiao2019control} the authors propose a more general form of high-order control barrier functions (HOCBFs). 

In practice, models used to design controllers are imperfect and this imperfection may lead to unsafe or even dangerous behavior. As a result, synthesizing a controller considering uncertainty is of great importance. Bounded disturbance is used to model such uncertainty as in \cite{taylor2020learning} and \cite{wang2020learning}, in which the time derivative of barrier functions are separated into the time derivative of the nominal barrier function and a remainder that can be approximated using neural networks. For systems driven by Gaussian-type  noise, stochastic differential equations against Brownian motions are usually used to characterize the effect of randomness. Previous studies on stochastic stability, given diffusion-type SDEs, have a wide variety of applications in verifying probabilistic quantification of safe set invariance \cite{kushner1967stochastic}. Investigations are focused on the worst-case safety verification utilizing SZCBFs regardless of the intensity of noise. The authors in \cite{santoyo2021barrier} proposed a method for synthesizing
polynomial state feedback controllers that achieve a
specified probability of safety based on the existing verification results. In connection to stochastic hybrid systems with more complex specifications, \cite{prajna2007framework} proposed a compositional framework for the construction of control barrier functions
for networks of continuous-time stochastic hybrid systems enforcing complex logic specifications expressed
by finite-state automata. However, the conservative quantification fails to be applied to high-order control systems due to the low-quality estimation. The authors in \cite{sarkar2020high} applied the pathwise strong safe set invariance certificate (generated by SRCBFs) from \cite{clark2019control} to high-order stochastic control systems. The SRCBF conditions are rather strong and effectively cancel the effects of diffusion and force the processes to stay invariant with probability 1. However, the above results admit unbounded control inputs, and hard constraints may cause failures of satisfying safety specifications.

\subsection{Contributions}
Motivated by the need to reduce potentially severe control constraints generated by SRCBF certificates\cite{clark2019control} in the neighborhood of the safety boundary (see Section \ref{sec:ex} for an illustrative example and Section \ref{sec:example1} for numerical comparisons), and improving the worst-case safety probability provided by SZCBF certificates,  we propose SCBFs as a middle ground to characterize safety properties for systems driven by Brownian motions in this paper. We show that SCBFs generate milder conditions compared to SRCBFs at the cost of sacrificing pathwise almost surely safety. The verification results, which provide a non-vanishing lower bound of safety probability for any finite time period, are still less conservative than widely used SZCBFs.

Unlike control systems with relative degree one, where optimal control schemes can be applied to synthesize finite-time almost-surely reachability/safety controller \cite[Chapter 5]{kushner1967stochastic} or even to characterize the probabilistic winning set of finite-time  reachability/safety with a priori probability requirement \cite{esfahani2016stochastic}, off-the-shelf optimal control schemes and numerical tools cannot be straightforwardly applied for stochastic control systems with high relative degree. Nonetheless, we make a first attempt to discover high-order SCBFs properties. 

The main contributions are summarized as follows. 
\begin{itemize}
    \item We propose a notion of stochastic control barrier functions (SCBFs) for safety-critical control. We show both theoretically and empirically that the proposed SCBFs achieve good trade-offs between mitigating the severe control constraints (potentially unbounded control inputs) and quantifying the worse-case safety probability. 
    \item Based on the less conservative worse-case safety probability, we extend our result to handle high relative degree safety constraints using high-order SCBFs and formally prove the worst-case safety probability in this case. 
    \item We validate our work in simulation and show that the proposed SCBFs with high relative degree have less control effort compared with SRCBFs and better safe probability compared with SZCBFs. 
\end{itemize}

\subsection*{Notation}

We denote the Euclidean space by $\R^n$ for $n>1$. Let $(\Omega,\mathscr{F}, \{\mathscr{F}_t\}_{t\geq 0}, \mathbb{P})$ be a filtered probability space. For any continuous-time stochastic processes $\{X_t\}_{t\geq 0}$ we use the shorthand notation $X:=\{X_t\}_{t\geq 0}$ instead, and denote $\langle X\rangle_t$ by the quadratic variation. We denote $\ub$ by a set of constrained control signals. In addition, let $\mathbb{P}_x$ be the probability measure of a stochastic process $X$ with the initial condition $X_0=x$ $\mathbb{P}$-a.s.; %\ymmark{let $\mathbb{P}_x^u$ be the probability measure of the process $X$ with the same initial condition under a control signal $u\in\ub$}.
correspondingly, we denote $\E^x$ by the expectation w.r.t. the probability measure $\mathbb{P}_x$ (i.e. $\mathbb{E}^x[f(X_t)]:=\mathbb{E}[f(X_t)|X_0=x]$).

 We say a function $\alpha:\,
\R_{\ge 0}\rightarrow\R_{\ge 0}$ belongs to class $\mathcal{K}$ if it is continuous, zero at zero, and strictly increasing. It is said to belong to $\mathcal{K}_{\infty}$ if it belongs to class $\mathcal{K}$ and is unbounded. For a given set $\mathcal{C}$, we refer $\mathcal{C}^\circ$ as the interior. 

\section{Preliminary and Problem Definition} \label{sec:definition}
\subsection{System Description}
Given a filtered probability space $(\Omega,\mathscr{F},\{\mathscr{F}_t\}, \mathbb{P})$ with a natrual filtration, a state space $\mathcal{X}\subseteq\R^n$, a (compact) set of control values $\U\subset \R^p$, consider a continuous-time stochastic process $X:[0,\infty)\times\Omega\rightarrow\X$ that solves the SDE
\begin{equation}\label{E:sys}
dX_t=(f(X_t)+g(X_t)u(t))dt+\sigma(X_t)dW_t,
\end{equation}
where $u:\R_{\geq 0}\rightarrow\U$ is a bounded measurable control signal; $W$ is a $d$-dimensional standard $\{\mathscr{F}\}_t$-Brownian motion; $f:\X\rightarrow \R^n$ is a nonlinear vector field; $g:\X\rightarrow\R^{n\times p}$ and $\sigma:\X\rightarrow\R^{n\times d}$ are smooth mappings. 

\begin{ass}\label{ass: usual}
We make the following assumptions on system \eqref{E:sys} for the rest of this paper:
\begin{itemize}
    \item[(i)] There is a $\xi\in\X$ such that $\mathbb{P}[X_0=\xi]=1$;
    \item[(ii)] The mappings $f,g,\sigma$ satisfy local Lipschitz continuity and a linear growth condition.
\end{itemize}
\end{ass}
\begin{definition}[Strong solutions]
A stochastic process $X$ is said to be a strong solution to \eqref{E:sys} if it satisfies the following integral equation
\begin{equation}
    X_t=\xi+\int_0^t(f(X_s)+g(X_s)u(s))ds+\int_0^t \sigma(X_s)dW_s, 
\end{equation}
where the stochastic integral is constructed based on the given Brownian motion $W$. 
\end{definition} 
\begin{remark}

\begin{itemize}
    \item[(i)]Under Assumption \ref{ass: usual}, the SDE \eqref{E:sys} admits a unique strong solution. 
    \item[(ii)] Weak solutions are in the sense that Brownian motions are constructed  posteriori for the stochastic integrals. For instance, by It\^{o}'s Lemma, the 2-norm $Y_t:=\|W_t\|_2$ of a $d$-dimensional Brownian motion $W$ satisfies the nonlinear SDE 
    \begin{equation}\label{E: example_strong}
        dY_t=\nabla_x\|W_t\|_2dW_t+\frac{d-1}{2Y_t}dt.
    \end{equation}
    Define $\tilde{W}_t:=\int_0^t\nabla_x\|W_t\|_2dW_t$. Then one can verify $\langle \tilde{W}\rangle_t=t$. By L\'evy's characterization, $\tilde{W}$ is again a Brownian motion, and therefore $Y$ also solves 
        \begin{equation}\label{E: example_weak}
dY_t=d\tilde{W}_t+\frac{d-1}{2Y_t}dt.
    \end{equation}
\end{itemize}
The solution $Y$ to \eqref{E: example_weak} is weak w.r.t. $\tilde{W}$. Even though any two solutions
(weak or strong) are identical in law, for future references,
we exclude the consideration of weak solutions in this paper  to guarantee that Lyapunov-type analysis on sample paths behaviors is based on the same Brownian motion. 
\end{remark}
\begin{definition} [Infinitesimal generator of $X_t$]
Let $X$ be the strong solution to \eqref{E:sys}, the infinitesimal generator $\Aa$ of $X_t$ is defined by 
\begin{equation}
    \mathcal{A}h(x)=\lim\limits_{t\downarrow 0}\frac{\E^x[h(X_t)]-h(x)}{t};\;\;x\in\R^n,
\end{equation}
where $h:\R^n\rightarrow \R$ is in a set $\mathcal{D}(\Aa)$ (called the domain of the operator $\mathcal{A}$) of functions  such that the limit exists at $x$.
\end{definition}
\begin{prop}[Dynkin]
Let $X$ solve \eqref{E:sys}. If $h\in C_0^2(\mathbb{R}^n)$ then $h\in\mathcal{D}(\mathcal{A})$ and 
\begin{equation}
    \mathcal{A}h(x)=\frac{\partial h}{\partial x}(f(x)+g(x)u(t))+\frac{1}{2}\sum\limits_{i,j}\left(\sigma\sigma^T\right)_{i,j}(x)\frac{\partial^2h}{\partial x_i\partial x_j}.
\end{equation}
\end{prop}
\begin{remark}
The solution $X$ to \eqref{E:sys} is right continuous and satisfies strong Markov properties, and for any finite stopping time $\tau$ and $h\in C_0^2(\R^n)$, we have the following Dynkin's formula
$$\E^\xi[h(X_\tau)]=h(\xi)+\E^\xi\left[\int_0^\tau \Aa h(X_s)ds\right],$$
and therefore
$$h(X_\tau)=h(\xi)+\int_0^t\Aa h(X_s)ds+\int_0^\tau \nabla_x h(X_s)dW_s.$$
The above is an analogue of the evolution of $h$ along trajectories
$$h(x(t))=h(\xi)+\int_0^t [L_fh(x(s))+L_gh(x(s))u(s)]ds$$
driven by deterministic dynamics $\dot{x}=f(x)+g(x)u$, where $L_fh=\nabla_x h(x)\cdot f(x)$ and $L_gh=\nabla_x h(x)\cdot g(x)$.
\end{remark}

\subsection{Set Invariance and Control}
In deterministic settings, a set $\mathcal{C}\subseteq \X$ is said to be invariant for a dynamical system $\dot{x}=f(x)$ if, for all $x(0)\in\mathcal{C}$, the solution $x(t)$ is well defined and $x(t)\in \mathcal{C}$ for all $t\geq 0$. As for stochastic analogies, we have the following probabilistic characterization of set invariance.
\begin{definition}[Probabilistic set invariance]
Let $X$ be a stochastic process $X$. A set $\mathcal{C}\subset\X$ is said to be invariant w.r.t. a tuple $(x, T, p)$ for $X$, where $x\in \mathcal{C}$, $T\ge 0$, and $p\in[0,1]$, if $X_0=x$ a.s. implies 
\begin{equation}
    \mathbb{P}_x[X_t\in\mathcal{C}, \;0\leq t\leq T]\geq p.
\end{equation}
Moreover, if $\mathcal{C}\subset\X$ is invariant w.r.t. $(x, T, 1)$ for all $x\in\mathcal{C}$ and $T\geq 0$, then $\mathcal{C}$ is strongly invariant for $X$.
\end{definition}

For stochastic dynamical systems with controls such as system \eqref{E:sys}, we would like to define similar probabilistic set invariance property for the controlled processes. Before that, we first define the
following concepts.
\begin{definition}[Control strategy]
A control strategy is a set-valued function
\begin{equation}
   \kappa:\X\rightarrow 2^{\U}. 
\end{equation}
\end{definition}

We use a boldface $\ub$ to indicate a set of constrained control signals. A special set of such signals is given by a control strategy as defined below. 

\begin{definition}[State-dependent control]\label{E: control constraint}
%A control constraint $\ub_\kappa$ is said to conform to a control strategy $\kappa$ if 
We say that a control signal $u$ conforms to a control strategy $\kappa$ for (\ref{E:sys}), and writes $u\in \ub_\kappa$, if 
\begin{equation}
    %\ub_\kappa(t)=\kappa(X_t),\quad\forall t\geq 0.
    u(t) \in \kappa(X_t),\quad\forall t\geq 0,
\end{equation}
where $X$ satisfies (\ref{E:sys}) with $u$ as input. The set of all control signals that confirm to $\kappa$ is denoted by $\ub_\kappa$. 
%A state-dependent control signal $u$ that conforms to $\kappa$ is an admissible control signal such that $u(t)\in\mathbf{u}_\kappa$ for all $t\geq 0$. 
\end{definition}

\begin{definition}[Controlled probabilistic invariance]
Given system \eqref{E:sys} and a set of control signals $\ub$, a set $\mathcal{C}\subset\X$ is said to be controlled invariant under $\ub$ w.r.t. a tuple $(\xi, T, p)$ for system \eqref{E:sys}, if for all $u\in\ub$, $\mathcal{C}$ is invariant w.r.t. $(x, T, p)$ for $X$, where $X$ is the solution to \eqref{E:sys} with $u$ as input. 
% if the following are satisfied:
% \begin{itemize}
%     \item[(i)]$\xi\in\mathcal{C}$; 
%     \item[(ii)]$p\in[0,1]$;
%     \item[(ii)]for all $u\in\ub$, the solution $X$ to \eqref{E:sys} remaining in set $\mathcal{C}$ has a probability estimation
% \begin{equation}
%     \mathbb{P}_\xi[X_t\in\mathcal{C}, \;0\leq t\leq T]\geq p.
% \end{equation}
% \end{itemize}
Similarly, $\mathcal{C}\subset\X$ is strongly controlled invariant under $\ub$ if $\mathcal{C}\subset\X$ is controlled invariant under $\ub$ w.r.t. $(\xi, T, 1)$ for all $\xi\in\mathcal{C}$ and $T\geq 0$. 
\end{definition}

\subsection{Problem Definition}
For the rest of this paper, we consider a safe set of the form
\begin{equation}\label{E: safeset}
    \mathcal{C}:=\{x\in\X: h(x)\geq 0\},
\end{equation}
where $h:\X\rightarrow\R$ is a high-order continuously differentiable function. We also define the boundary and interior of $\mathcal{C}$ explicitly as below
\begin{equation}
    \partial\mathcal{C}:=\{x\in\X: h(x)= 0\},
\end{equation}
\begin{equation}
    \mathcal{C}^\circ:=\{x\in\X: h(x)> 0\}.
\end{equation}
\begin{prob}[Probabilistic set invariance control] \label{prob: main}
Given a compact set $\mathcal{C}\subset\X$ defined in \eqref{E: safeset}, a point $\xi\in\mathcal{C}^\circ$, and a tuple $(\xi,T,p)$, design a control strategy $\kappa$ such that under $\ub_\kappa$, the interior $\mathcal{C}^\circ$ is controlled invariant w.r.t. $(\xi,T,p)$  for the resulting strong solutions to \eqref{E:sys}.
\end{prob}
\section{Safe-critical Control Design via Barrier Functions}
In this section, we propose stochastic barrier certificates that can be used to design a
control strategy $\kappa$ for Problem \ref{prob: main}. Before proceeding, it is necessary
to review (stochastic) control barrier functions to interpret (probabilistic) 
set invariance. Note that we consider the safe set as constructed in \eqref{E: safeset}, where the function $h$ is given a priori. 
\subsection{Stochastic Reciprocal and Zeroing Barrier Functions}
Similar to the terminology for deterministic cases \cite{ ames2016control}, we introduce the construction of stochastic control barrier functions as follows.
\begin{definition}[SRCBF]
A function $B:\mathcal{C}^\circ\rightarrow \R$ is called a stochastic reciprocal control barrier function (SRCBF) for system \eqref{E:sys} if $B\in\mathcal{D}(\mathcal{A})$ and satisfies the following properties:
\begin{itemize}
    \item[(i)] there exist class-$\mathcal{K}$ functions $\alpha_1,\alpha_2$ such that for all $x\in\X$ we have
    \begin{equation}
        \frac{1}{\alpha_1(h(x))}\leq B(x)\leq \frac{1}{\alpha_2(h(x))};
    \end{equation}
    \item[(ii)] there exists a class-$\mathcal{K}$ function $\alpha_3$ such that
    \begin{equation}\label{E: rcbf strate}
        \inf\limits_{u\in\U}[\mathcal{A}B(x)-\alpha_3(h(x))]\leq 0.
    \end{equation}
\end{itemize}
We refer to the control strategy generated by \eqref{E: rcbf strate} as 
\begin{equation}
    \upsilon(x):=\{u\in\U: \mathcal{A}B(x)-\alpha_3(h(x))\leq 0\}
\end{equation}
and the corresponding control constraint  as $\ub_{\upsilon}$ (see in Definition \ref{E: control constraint}).
\end{definition}
\begin{prop}[\cite{clark2019control}]
Suppose that there exists an SRCBF for system \eqref{E:sys}. If $u(t)\in\ub_\upsilon$, then for all $t\geq 0$ and  $X_0=\xi\in\mathcal{C}^\circ$, we have $\mathbb{P}_\xi[X_t\in\mathcal{C}^\circ]=1$ for all $t\geq 0$.
\end{prop}
\begin{remark}
The result admits a $\mathbb{P}$-a.s. controlled invariant set for the marginals of $X$, and is easily extended to a pathwise $\pp$-a.s. controlled set invariance. Note that the strong solution is right continuous.
Let $\{t_n,\;n=1,2,...\}$ be the set of all rational numbers in $[0,\infty)$, and put 
$$\Omega^*:=\bigcap\limits_{1\leq n<\infty}\{\omega: X_{t_n}\in\mathcal{C}^\circ\},$$
then $\Omega^*\in\mathscr{F}$ (a $\sigma$-algebra is closed w.r.t. countable intersections). 
Since $\mathbb{Q}$ is dense in $\mathbb{R}$, $X$ is right continuous, and $h$ is continuous, we have 
$$\Omega^*:=\{\omega: X_{t}\in\mathcal{C}^\circ,\;\forall t \in[0,\infty)\}.$$
Note that  $\mathbb{P}_\xi[\Omega^*]\equiv 1$ from the marginal result. 
\end{remark}
\begin{definition}[SZCBF]
A function $B:\mathcal{C}\rightarrow \R$ is called a stochastic zeroing control barrier function (SZCBF) for system \eqref{E:sys} if $B\in\mathcal{D}(\mathcal{A})$ and 
\begin{enumerate}
    \item[(i)] $B(x)\geq 0$ for all $x\in\mathcal{C}$;
    \item[(ii)] $B(x)<0$ for all $x\notin \mathcal{C}$;
    \item[(iii)] there exists an extended $\mathcal{K}_{\infty}$ function $\alpha$ such that
\begin{equation}\label{E: zcbf certificate}
    \sup_{u\in\U}{[\mathcal{A}B(x)}+\alpha({h(x)})]\geq 0.
\end{equation}
\end{enumerate}
We refer the control strategy generated by \eqref{E: zcbf certificate} as 
\begin{equation}
    \upsilon(x):=\{u\in\U: {\mathcal{A}B(x)}+\alpha({h(x)})\geq 0\}
\end{equation}
and the corresponding set of constrained control signals  as $\ub_{\upsilon}$.
\end{definition}
\begin{prop}[Worst-case probabilistic quantification]
If the mapping $h$ is an SZCBF with linear function $kx$ as the class-$\mathcal{K}$ function, where $k\geq 0$,  then the control strategy is given as $\upsilon(x)=\{u\in\U: \mathcal{A}h(x)+kh(x)\geq 0\}$. Let $c=\sup_{x\in\mathcal{C}}h(x)$ and  $X_0=\xi\in\mathcal{C}^\circ$, then under any $u\in\ub$ %the set of constrained control signals $\ub_\upsilon$, 
we have the following worst-case probability estimation:
\begin{equation}
    \mathbb{P}_{\xi}\left[X_t\in\mathcal{C}^\circ,\;0\leq t\leq T\right]\geq \left(\frac{h(\xi)}{c}\right)e^{-cT}.
\end{equation}
\end{prop}
\begin{proof}
Let $s=c-h(\xi)$ and $V(x)=c-h(x)$, then $V(x)\in [0,c]$ for all $x\in \mathcal{C}$. It is clear that $\mathcal{A}V(x)=-\mathcal{A}h(x)$ and $\upsilon(x)=\{u\in\U: \inf\limits_{u\in\U}[\mathcal{A}V(x)+kh(x)]\leq 0\}$. For $u(t)\in\ub_\upsilon$ for all $t\in[0,T]$, we have
$$\mathcal{A}V(x)\leq -kV(x)+kc.$$
By \cite[Theorem 3.1]{kushner1967stochastic}, 
 \begin{equation}
     \begin{split}
        \mathbb{P}_{\xi}\left[\sup\limits_{t\in[0,T]}V(X_t)\geq c\right]\leq 1-\left(1-\frac{s}{c}\right)e^{-cT}. 
     \end{split}
 \end{equation}
The result follows directly after this.
\end{proof}

In proposing stochastic control barrier functions for high-order control systems, the above SRCBF and SZCBF are building blocks. The authors in \cite{sarkar2020high} constructed high-order SRCBF and have found the sufficient conditions to guarantee pathwise set invariance with probability 1.
While the results seem strong, they come with significant costs. At the safety boundary, the control inputs need to be unbounded (as shown in the motivating example below and in the numerical experiments in Section \ref{Sec: sim}). On the other hand, the synthesis of controller for a high-order system via SZCBF is with mild constraints. The trade-off is that the probability estimation of set invariance is of low quality (note that the worst-case probability estimation using first-order barrier function is already lower bounded by a small value over a relatively long time period). We propose high-order stochastic control barrier functions in subsection C in order to reduce the high control efforts and improve the worst-case quantification.  Before proceeding, we illustrate the above motivation through examples.

\subsection{A Motivating Example}\label{sec:ex}

In this section, we will discuss how reciprocal control barrier functions (RCBFs) and SRCBFs perform differently around the boundary of the safe set. We show that for deterministic systems, RCBFs can guarantee safety with bounded control while for stochastic systems, SRCBFs require unbounded control in order to keep systems safe. We use the following two simple one-dimensional systems:
\begin{equation*}
    \dot{x}=x+u,
\end{equation*}
and
\begin{equation*}
    dx=x+u+\sigma dW,
\end{equation*}
for the comparison. Suppose that our safe set is $\{x\in\mathbb{R}\big|x<1\}$ so that we can use $h(x)=1-x$. Accordingly, a RCBF for the deterministic system is  $B=\frac{1}{h}$. We choose $\gamma=1$ as in \cite{ames2016control} and \cite{clark2019control} and, as a result, the condition using the RCBF is 
\begin{equation*}
    \begin{split}
        L_fB(x)+L_gB(x)u&=\frac{1}{h^2}(x+u)\leq h,\\
        u&\leq (1-x)^3-x.\\
    \end{split}
\end{equation*}
It means that we can control the system safely using a control bounded by $(1-x)^3-x$ when $x\to 1$. However, for the stochastic system, we have $\frac{\partial^2 B}{\partial x^2}=\frac{2}{h^3}$. Then the SRCBF condition is 
\begin{equation*}
    \begin{split}
        \mathcal{A}B&=\frac{1}{h^2}(x+u)+\frac{2\sigma^2}{h^3}\leq h,\\
        u&\leq h^3-x-\frac{2\sigma^2}{h}.
    \end{split}
\end{equation*}
As $x$ approaches 1, the control approaches $-\infty$. This implies that in order to guarantee safety, we requires an unbounded control around the boundary of the safe set for stochastic systems, which can be difficult to satisfy for some practical applications, as shown in Section~\ref{sec:example1}. \\

\subsection{High-order Stochastic Control Barrier Functions}
To obtain non-vanishing worst-case probability estimation (compared to SZCBF), we propose a safety certificate via a stochastic Lyapunov-like control barrier function \cite{kushner1967stochastic}.
\begin{definition}[Stochastic control barrier functions] A continuously differentiable function $B:\R^n\rightarrow \R$ is said to be a stochastic control barrier function (SCBF) if $B\in\mathcal{D}(\mathcal{A})$ and the following conditions are satisfied:
\begin{enumerate}
    \item[(i)] $B(x)\geq 0$ for all $x\in\mathcal{C}$;
    \item[(ii)] $B(x)<0$ for all $x\notin \mathcal{C}$;
    \item[(iii)] $
        \sup\limits_{u\in\U}\mathcal{A}B(x)\geq 0
    $.
\end{enumerate}
We refer the control strategy generated by (iii) as 
\begin{equation}
    \upsilon(x):=\{u\in\U: \mathcal{A}B(x)\geq 0\}
\end{equation}
and the corresponding set of constrained control signals  as $\ub_{\upsilon}$.
\end{definition}
\begin{remark}
Condition (iii) of the above definition is an analogue of $\sup\limits_{u\in\U}[L_fB(x)+L_gB(x)u]\geq 0$ for the deterministic settings. The consequence is such that $\mathbb{E}^\xi[B(x)]\geq B(x)$ for all $x\in\mathcal{C}$. A relaxation of condition (iii) is given in the deterministic case such that the set invariance can still be guaranteed \cite{wang2020learning},
$$\sup\limits_{u\in\U}[L_fB(x)+L_gB(x)u]\geq -\alpha(B(x)), $$
where $\alpha$ is a class-$\mathcal{K}$ function. If $\alpha(x)=kx$ where $k>0$, under the stochastic settings, the condition formulates an SZCBF and provides a much weaker quantitative estimation of the lower bound of satisfaction probability. In comparison with SZCBF, we provide the worst-case quantification in the following proposition.
\end{remark}
\begin{prop}\label{prop: main}
Suppose the mapping $h$ is an SCBF  with the corresponding control strategy  $\upsilon(x)$. Let $c=\sup_{x\in\mathcal{C}}h(x)$ and $X_0=\xi\in\mathcal{C}^\circ$, then under the set of constrained control signals $\ub_\upsilon$, we have the following worst-case probability estimation:
$$
\mathbb{P}_{\xi}\left[X_t\in\mathcal{C}^\circ,\;0\leq t<\infty\right]\geq \frac{h(\xi)}{c}.
$$
\end{prop}
\begin{proof}
Let $V=c-h(x)$, then $V(x)\geq 0$ for all $x\in\X$ and $\mathcal{A}V\leq 0$. Let $s=c-h(\xi)$ then $s\leq c$ by definition.  
\iffalse
Let $V(x)=c-h(x)$, then $V(x)\geq 0$  
\fi
The result for every finite time interval $t\in[0,T]$ is followed by \cite[Lemma 2.1]{kushner1967stochastic},
$$\mathbb{P}_{\xi}\left[X_t\in\mathcal{C}^\circ,\;0\leq t<T\right]\geq 1-\frac{s}{c}.$$
The result follows by letting $T\rightarrow\infty$.
\end{proof}
\begin{definition}
A function $B:\X\rightarrow \R$ is called a stochastic control barrier function with relative degree $r$ for system \eqref{E:sys} if $B\in\mathcal{D}(\mathcal{A}^r)$, and  $\mathcal{A}\cdot\mathcal{A}^{r-1}h(x)\neq 0$ as well as  $\mathcal{A}\cdot\mathcal{A}^{j-1}h(x)=0$ for $j=1,2,\ldots,r-1$ and $x\in \mathcal{C}$.
\end{definition}

If the system \eqref{E:sys} is an $r^{\text{th}}$-order stochastic control system, to steer the process $X$ to satisfy probabilistic set invariance w.r.t. $\mathcal{C}$, we recast the mapping $h$ as an SCBF with relative degree $r$. For $h\in\mathcal{D}(\mathcal{A}^r)$, 
we define a series of functions $b_0,b_j:\X\to\mathbb{R}$ such that for each $j=1,2,\dots,r$ $b_0,b_j\in\mathcal{D}(\mathcal{A})$  and
\begin{equation}\label{eq:hocbf}
\begin{split}
    b_0(x)&=h(x),\\
    b_j(x)&=\mathcal{A}\cdot \mathcal{A}^{j-1}b(x).\\
    %&\vdots\\
    %b_r(x)&=\dot{b}_{r-1}+\alpha_r(b_{r-1}(x)),\\
\end{split}
\end{equation}
 We further define the corresponding superlevel sets $\mathcal{C}_j$ for $j=1,2,\dots,r$ as
\begin{equation}\label{eq:hoset}
\begin{split}
  \mathcal{C}_j&=\{x\in\mathbb{R}^n:b_{j}(x)\geq 0\}.\\
  %\mathcal{C}_1&=\{x\in\mathbb{R}^n:b_1(x)\geq0\},\\
  %&\vdots\\
  %\mathcal{C}_{r-1}&=\{x\in\mathbb{R}^n:b_r(x)\geq0\}.
\end{split}
\end{equation}
\iffalse
For future references, we define 
\begin{enumerate}
    \item $c_j=:\sup_{x\in\X}b_j(x)$;
    %\item $r_j:=c_j-b_j(\xi)$;
    %\item $p_j^T:=\mathbb{P}_{X_0}[X_t\in\mathcal{C}_j^\circ,\;0\leq t\leq T]$; $\hat{p}_j^T:=\mathbb{P}_{X_0}[X_t\in\mathcal{C}_j^\circ,\;0\leq t\leq T\;|\;\mathcal{A}b_j\leq 0]$
    \item $p_j:=\mathbb{P}_{X_0}[X_t\in\mathcal{C}_j^\circ,\;0\leq t<\infty]$; $\hat{p}_j:=\mathbb{P}_{X_0}[X_t\in\mathcal{C}_j^\circ,\;0\leq t<\infty\;|\;\mathcal{A}b_j\geq 0]$
\end{enumerate}  
where $j=1,2,\dots,r$.
\fi
\begin{thm}\label{thm}
If the mapping $h$ is an SCBF with relative degree $r$, the corresponding control strategy is given as $\upsilon(x)=\{u\in\U: \mathcal{A}^rh(x)\geq 0\}$. Let $c_j=:\sup_{x\in\mathcal{C}_j}b_j(x)$ for each $j=0,1,...,r$ and $X_0=\xi\in\bigcap_{j=0}^r \mathcal{C}_j^\circ$. Then under the set of constrained control signals $\ub_\upsilon$, we have the following worst-case probability estimation:
$$\mathbb{P}_\xi[X_t\in\mathcal{C}^\circ,\;0\leq t<\infty]\geq \prod\limits_{j=0}^{r-1} \frac{b_j(\xi)}{c_j}.$$
\end{thm}
\begin{proof}
We introduce the notations 
$p_j:=\mathbb{P}_\xi[X_t\in\mathcal{C}_j^\circ,\;0\leq t<\infty]$ ,  $\hat{p}_j:=\mathbb{P}_\xi[X_t\in\mathcal{C}_j^\circ,\;0\leq t<\infty\;|\;\mathcal{A}b_j\geq 0]$.% and $s_j:=c_j-b_j(\xi)$for $j=1,2,...,r$.

The control signal $u(t)\in\{u\in\U: \mathcal{A}^rb(x)\geq 0\}$ for all $t\geq 0$ provides $\mathcal{A}b_{r-1}\geq 0$. By Proposition \ref{prop: main},
\begin{equation}
    \begin{split}
     p_{r-1} &=\mathbb{P}_\xi[X_t\in\mathcal{C}_{r-1}^\circ,\;0\leq t<\infty]\\
     & =\mathbb{P}_\xi[X_t\in\mathcal{C}_{r-1}^\circ,\;0\leq t<\infty\;|\;\mathcal{A}b_{r-1}\geq 0]\\
     &=\hat{p}_{r-1}.
    \end{split}
\end{equation}

For $j=0,1,...,r-2$, and $0\leq t<\infty$, we have the following recursion:
\begin{equation}\label{E: recursion}
\begin{split}
    & p_j=\mathbb{P}_\xi[X_t\in\mathcal{C}_{j}^\circ]\\ & =\mathbb{E}^{\xi}[\mathds{1}_{\{X_{t}\in C_{j}^\circ\}}\mathds{1}_{\{X_{t}\in C_{j+1}^\circ\}}]+\mathbb{E}^{\xi}[\mathds{1}_{\{X_{t}\in C_{j}^\circ\}}\mathds{1}_{\{X_{t}\notin C_{j+1}^\circ\}}]\\
    & =\mathbb{P}_\xi[X_{t}\in C_{j}^\circ\;|\mathds{1}_{\{X_{t}\in C_{j+1}^\circ\}}]\cdot\mathbb{P}_\xi[X_{t}\in C_{j+1}^\circ]\\
    & +\mathbb{P}_\xi[X_{t}\in C_{j}^\circ\;|\mathds{1}_{\{X_{t}\notin C_{j+1}^\circ\}}]\cdot\mathbb{P}_\xi[X_{t}\notin C_{j+1}^\circ],
\end{split}
\end{equation}
where we have used shorthand notations $\{X_t\in\mathcal{C}_{j}\}:=\{X_t\in\mathcal{C}_{j},\;0\leq t<\infty\}$ and $\{X_{t}\notin C_{j+1}\}:=\{X_{t}\notin C_{j+1}\;\text{for\;some}\;0\leq t<\infty\}$.
Indeed, we have $$ X_t=X_t\mathds{1}_{\{X_{t}\in C_{j+1}^\circ\}}+X_t\mathds{1}_{\{X_{t}\notin C_{j+1}^\circ\}},$$
then
$$ \mathbb{E}[X_t]=\mathbb{E}[X_t\mathds{1}_{\{X_{t}\in C_{j+1}^\circ\}}]+\mathbb{E}[X_t\mathds{1}_{\{X_{t}\notin C_{j+1}^\circ\}}],$$
and \eqref{E: recursion} follows. Note that
\begin{equation}
    \begin{split}
      &\mathbb{P}_\xi[X_{t}\in C_{j}^\circ\;|\mathds{1}_{
      \{X_{t}\in C_{j+1}^\circ\}}]\\
      &\geq \mathbb{P}_\xi[X_{t}\in C_{j}^\circ\;| \mathcal{A}b_{j}\geq 0]=\hat{p}_j,
    \end{split}
\end{equation}
and therefore 
$$\mathbb{P}_\xi[X_{t}\in C_{j}^\circ\;|\mathds{1}_{\{X_{t}\in C_{j+1}^\circ\}}]\cdot\mathbb{P}_\xi[X_{t}\in C_{j+1}^\circ]\geq \hat{p}_jp_{j+1}.$$
 %Now, let $v_j=-b_j$, then $\mathcal{A}v_j<0$ and $v_j\in[0,c]$ for $x\in\mathcal{C}_j$ and $v_j<0$ otherwise. 

 Now define stopping times $\tau_j=\inf\{t: b_j(X_{t})\leq 0\}$ for $j=0,1,...,r-2$, then $b_j(X_{t\wedge \tau_j})\geq 0$ a.s.. In addition,
 $$X_{t\wedge\tau_j}=\mathds{1}_{\{\tau_j\leq \tau_{j+1}\}}X_{t\wedge\tau_j}+\mathds{1}_{\{\tau_j> \tau_{j+1}\}}X_{t\wedge\tau_j}$$
 
 Assume the worst scenario, which is for all $t\geq \tau_{j+1}$, we have $\mathcal{A}b_j\leq 0$. On $\{\mathcal{A}b_j<0\}\cap \{\tau_j> \tau_{j+1}\}$, we have $b_j(X_{\tau_{j+1}})>0$ and $\mathbb{E}^{\xi}[b_j(X_{t\wedge\tau_j})]\leq b_j(X_{\tau_{j+1}})-\int_{\tau_{j+1}}^t\varepsilon(s)ds$ for some $\varepsilon:\mathbb{R}_{\geq 0}\rightarrow \mathbb{R}_{>0}$ and $t\geq \tau_{j+1}$. Therefore, the process $b_j(X_{t\wedge\tau_j})$  is a nonnegative supermartingale and 
 $\mathbb{P}_\xi[\sup_{T\leq t<\infty}b_j(X_{t\wedge\tau_j})\geq \lambda]\leq \frac{b_j(X_{\tau_{j+1}})-\int_{\tau_{j+1}}^T\varepsilon(s)ds}{\lambda}$ by Doob's supermartingale inequality. For any $\lambda>0$, we can find a finite $T\geq \tau_{j+1}$ such that $\mathbb{P}_\xi [\sup_{T\leq t<\infty}b_j(X_{t\wedge\tau_j})\geq \lambda]=0$. Since $\lambda$ is arbitrarily selected, we must have $\mathbb{P}_\xi[\sup_{T\leq t<\infty}b_j(X_{t\wedge\tau_j})>0]=0$, which means $\tau_j$ is triggered within finite time.
 \iffalse
 If $\{\mathcal{A}b_j\leq 0\}$, by \cite[Theorem 3-2]{kushner1967stochastic}, the process $X$ converges with probability one. Since $\mathbb{P}_\xi[b_j\geq c_j, \;0\leq t<\infty]=0$, under the above assumption the only limit point is $\{x: \mathcal{A}b_j(x)=0\}$. Now, 
 given $\{\mathcal{A}b_j<0\}$, the process is attracted to some point outside the boundary (including $\{x: b_j(x)=-\infty\}$) with probability one. 
 \fi
 \iffalse
 there exists a $\tau<\infty$ such that all sample paths on $\{X_t\in\mathcal{C}_j\}$ converge to $\{x:v_j(x)=0\}$  with probability less than $\hat{p}_j^T$ as $t\rightarrow \tau$. This means $\{X_t\notin\mathcal{C}_j^\circ\;\text{for\;some}\;0\leq t\leq \tau\}$ given $\{\mathcal{A}b_j<0\}$ possesses a probability higher than $1-\hat{p}_j^T$. 
 \fi
 On the other hand, on $\{\mathcal{A}b_j<0\}\cap \{\tau_j\leq  \tau_{j+1}\}$, $\tau_j$ has been already triggered. Therefore, $\{X_t\notin\mathcal{C}_j^\circ\;\text{for\;some}\;0\leq t<\infty\}$ a.s. given $\{\mathcal{A}b_j<0\}$.  Hence,
 \begin{equation}
     \begin{split}
         &\mathbb{P}_\xi[X_{t}\in C_{j}^\circ\;|\mathds{1}_{\{X_{t}\notin C_{j+1}^\circ\}}]\\
         &\geq \mathbb{P}_\xi[X_{t}\in C_{j}^\circ\;|\mathds{1}_{\{X_{t}\notin C_{j+1}^\circ,\;\forall t\geq \tau_{j+1}\}}]\\
         &\geq \mathbb{P}_\xi[X_{t}\in C_{j}^\circ\;|\mathds{1}_{\{\mathcal{A}b_j<0,\;\forall t\geq \tau_{j+1}\}}]=0
     \end{split}
 \end{equation}
 \iffalse
 Similarly, for $\mathcal{A}b_j=0$, we can show $b_j(X_t)$ is a martingale, so is $b_j(X_{t\wedge \tau_j})$. Using martingale convergence theorem, the probability of $b_j(X_{t\wedge \tau_j})$ reaching $0$ before any $y>b_j(X_{\tau_{j+1}})$ is $1-\frac{b_j(X_{\tau_{j+1}})}{y}$. \fi %Combine the above, we have $\mathbb{P}_{X_0}[\{X_{t}\in C_{j}^\circ\}\;|\mathds{1}_{\{X_{t_0}\notin C_{j+1}^\circ\}}]\cdot\mathbb{P}_{X_0}[\{X_{t_0}\notin C_{j+1}^\circ\}]\geq (1-\hat{p}_j)(1-p_{j+1})$. Hence, for $0\leq t\leq \tau$
 %$$p_j\geq \hat{p}_jp_{j+1}+(1-\hat{p}_j)(1-p_{j+1}).$$
 Combining the above, for $j=0,1,...,r-2$, we have  
 $$p_j\geq \hat{p}_jp_{j+1},$$
 and ultimately $p_0\geq \prod\limits_{j=0}^{r-1}\hat{p}_j$. %\prod\limits_{j=0}^{r-1} \frac{b_j(\xi)}{c_j}$.
\end{proof}
\begin{remark}\label{rem: thm}
The above result estimates the lower bound of the safety probability given the constrained control signals $\ub_\upsilon$. Based on recursion relation \eqref{E: recursion} from the proof, we can easily drop the last term and still obtain the same result. However, we made a detour to argue that under some extreme conditions the worst case may happen. Indeed, we have assumed that for all $t\geq \tau_{j+1}$, we have $\mathcal{A}b_j\leq 0$. This conservative assumption  covers all the conditions in which the downward flow ($\mathcal{A}b_j\leq 0$) maintains for enough time before the sample paths cross back to $\mathcal{C}_{j+1}^\circ$ (i.e. $\mathcal{A}b_j>0$), such that within finite time $X$ will cross the boundary of $\mathcal{C}_j$.

Another implicit condition may cause the worst-case lower bound as well, that is when $\bigcup_{j=0}^{r-1}\{\mathcal{A}b_j=0,\;0\leq t<\infty\}$ is a $\mathbb{P}_\xi$-null set. This, however, is practically possible since the controller indirectly influences the value of $\mathcal{A}b_j$ for all $j<r$, the strong invariance of the level set $\{\mathcal{A}b_j=0\}$ is not guaranteed using QP scheme. On the other hand, a non-zero probability of $\{\mathcal{A}b_j=0\}$ for any $j<r$ makes a compensate for the lower bound estimation. On $\{\mathcal{A}b_j=0\}$, the process $\{b_j(X_{t\wedge\tau_j)}\}_{t\geq 0}$ is a lower bounded martingale and therefore converges with probability 1. We can estimate (using standard exit-time problem arguments) that $\{b_j(X_{t\wedge\tau_j})\}_{t\geq 0}$ reaching $0$ within a finite time has a probability $1-\frac{b_j(X_{\tau_{j+1}})}{c_j}$.
\end{remark}
\begin{remark}\label{rem: controller}
A nice selection of controller is to implicitly reduce the total time a sample path spends in $\{\mathcal{A}b_j\leq 0\}$. Since by the supermartingale property of $\{b_j(X_{t\wedge\tau_j})\}_{t\geq 0}$, one can estimate for any $\lambda>0$,
$$\mathbb{P}_\xi\left[\sup\limits_{\infty>t\geq T}b_j(X_{t\wedge \tau})\geq \lambda\right]\leq \frac{b_j(X_T)}{\lambda}.$$
The estimation depends on the solution of $\dot{b}_j\leq 0$, of which a better controller may implicitly increase the convergence rate(see examples in \cite[Theorem 3-4]{kushner1967stochastic}, such that the backward flow vanishes within a short time period. However, this is a challenging task by only steering the bottom-level flow, which in turn gives us a future research direction.
\end{remark}

\section{Simulation Results}\label{Sec: sim}
In this section, we use two examples to validate our result. We show that the proposed SCBFs have smaller control effort compared to SRCBFs and higher safe probability compared to SZCBFs. 

\subsection{Example 1}\label{sec:example1}
In the first example, we use an automatic cruise control example as in \cite{clark2019control}  and \cite{ames2016control}. The model is given by the following three-dimensional system:
\begin{equation*}
    d\begin{bmatrix}
    x_1\\
    x_2\\
    x_3\\
    \end{bmatrix}=\begin{bmatrix}
    -F_r(x)/M\\
    0\\
    x_2-x_1\\
    \end{bmatrix}dt+\begin{bmatrix}
    1/M\\
    0\\
    0\\
    \end{bmatrix}udt+\Sigma dW,
\end{equation*}
where $x_1$ and $x_2$ denote the velocity of the following vehicle and leading vehicle, respectively, and $x_3$ is the distance between two vehicles and 
\begin{equation*}
    \Sigma=\begin{bmatrix}
    \sigma_1&0&0\\
    0&0&0\\
    0&0&\sigma_2\\
    \end{bmatrix}.
\end{equation*} 
The aerodynamic drag is $F_r(x)=f_0+f_1x_1+f_2x_1^2$ with $f_0=0.1$, $f_1=5$, $f_2=0.25$ and the mass of the vehicle is $M=1650$. The initial state is chosen as $[x_1, x_2, x_3]=[18, 10, 150]^T$ and $W$ is a three-dimensional Brownian motion representing uncertainty in the velocity of following vehicle and the distance of two vehicles. The goal of the following vehicle is to achieve a desired velocity $x_d=22$ while keeping the collision constraint $h(x)=x_3-\tau x_1>0$. We use a Lyapunov function $V(x)=(x_1-x_d)^2$ to control the velocity of the following vehicle. We use $B(x)=\frac{1}{h(x)}$ as the SRCBF and $h(x)$ as the SCBF. We solve QP problems using SRCBF as
\begin{equation*}
    \begin{split}
        &[u^*, \delta^*]=\argmin_{u,\delta} \frac{1}{2}(u^2+\delta^2)\quad\text{s.t.}\\
        &\frac{\partial V(x)}{\partial x}(f(x)+g(x)u)+\frac{1}{2}\Tr\bigg({\Sigma^T\frac{\partial^2 V(x)}{\partial x^2 }\Sigma}\bigg)\leq\delta,\\
         &\frac{\partial B(x)}{\partial x}(f(x)+g(x)u)+\frac{1}{2}\Tr\bigg({\Sigma^T\frac{\partial^2 B(x)}{\partial x^2 }\Sigma}\bigg)\leq\frac{\gamma}{B(x)},
    \end{split}
\end{equation*}

and SCBF as

\begin{equation*}
    \begin{split}
        &[u^*, \delta^*]=\argmin_{u,\delta} \frac{1}{2}(u^2+\delta^2)\quad\text{s.t.}\\
        &\frac{\partial V(x)}{\partial x}(f(x)+g(x)u)+\frac{1}{2}\Tr\bigg({\Sigma^T\frac{\partial^2 V(x)}{\partial x^2 }\Sigma}\bigg)\leq\delta,\\
         &\frac{\partial h(x)}{\partial x}(f(x)+g(x)u)+\frac{1}{2}\Tr\bigg({\Sigma^T\frac{\partial^2 h(x)}{\partial x^2 }\Sigma}\bigg)\geq 0.\\
    \end{split}
\end{equation*}

We present the simulation results by showing the scaled control value $u/(Mg)$ and control effort $J=u^2$ for SCBF and SRCBF as in Figure~\ref{fig:u} and Figure~\ref{fig:J}. For simplicity, we choose noise level to be $\sigma_1=\sigma_2=1$. In both figures, the red curves are for the SRCBF and the blue curves are for the SCBF. From the figures, we can find out that SRCBF requires an impulse control signal to make the system safe when approaching the boundary. The peak value of $J$ using SRCBF is almost $1.75e^{9}$ while it is less than $0.1e^9$ for that of using SCBFs. In practice, this implies that we need a very large acceleration in order to keep the system staying within the safe set. As a result, we bound the scaled control to be $u/(Mg)>-0.5$ and show the safe probability as in Table~\ref{tab}. The table shows that for unbounded control, the safe probability using SRCBF is 90\% compared to 70\% by using SCBF. However, the safe probability drops to only 25\% when we use SRCBF while the safe probability is 65\% for SCBF under bounded control input. As a result, we can see that the safety probability obtained by SCBF is more robust to saturation of control inputs.

\begin{figure}[htbp]
	\centering
	\includegraphics[width=1.1\linewidth]{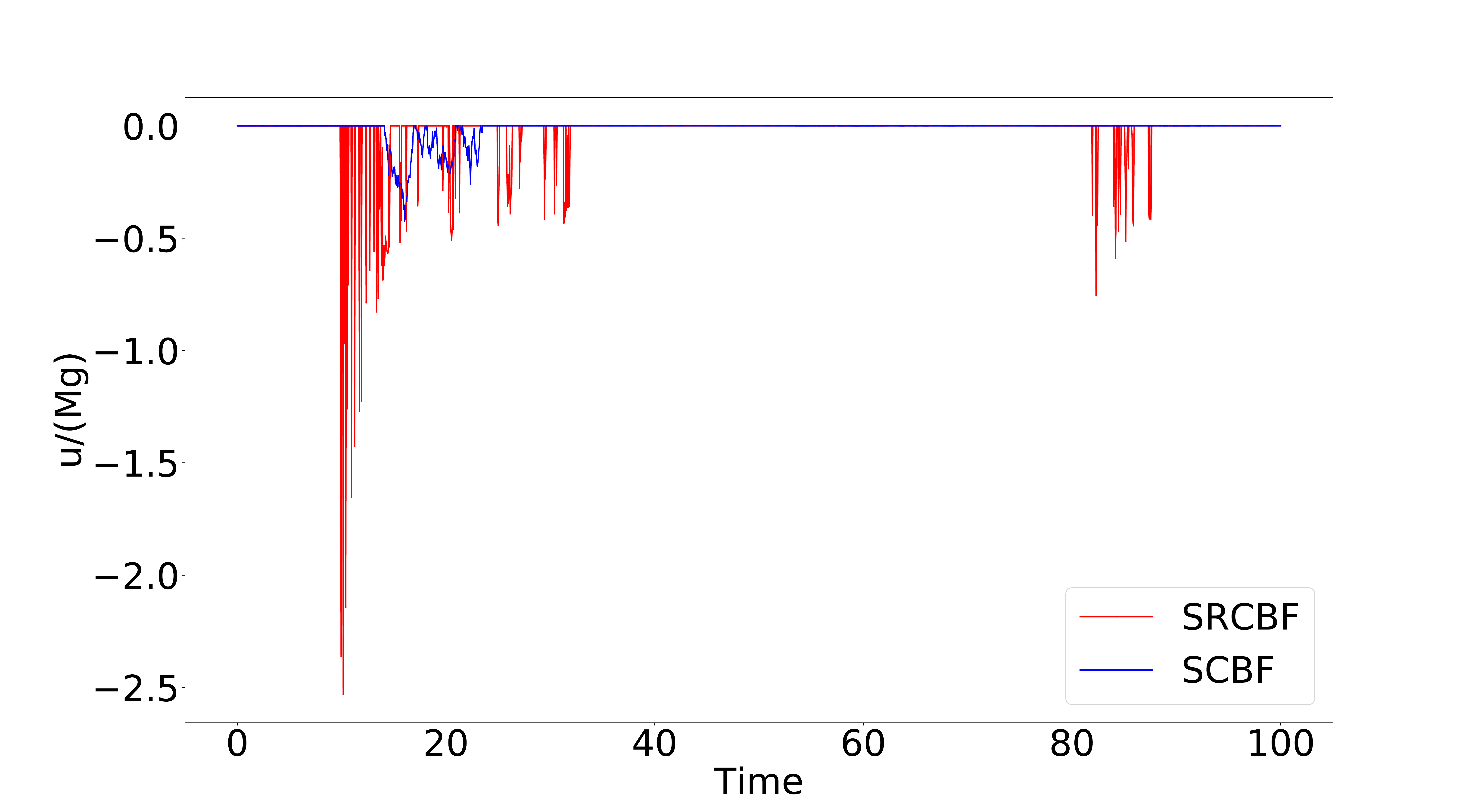}
	\caption{Plot of scaled control $u/(Mg)$ for Example 1.}
	\label{fig:u}
\end{figure}%

\begin{figure}[htbp]
	\includegraphics[width=1.1\linewidth]{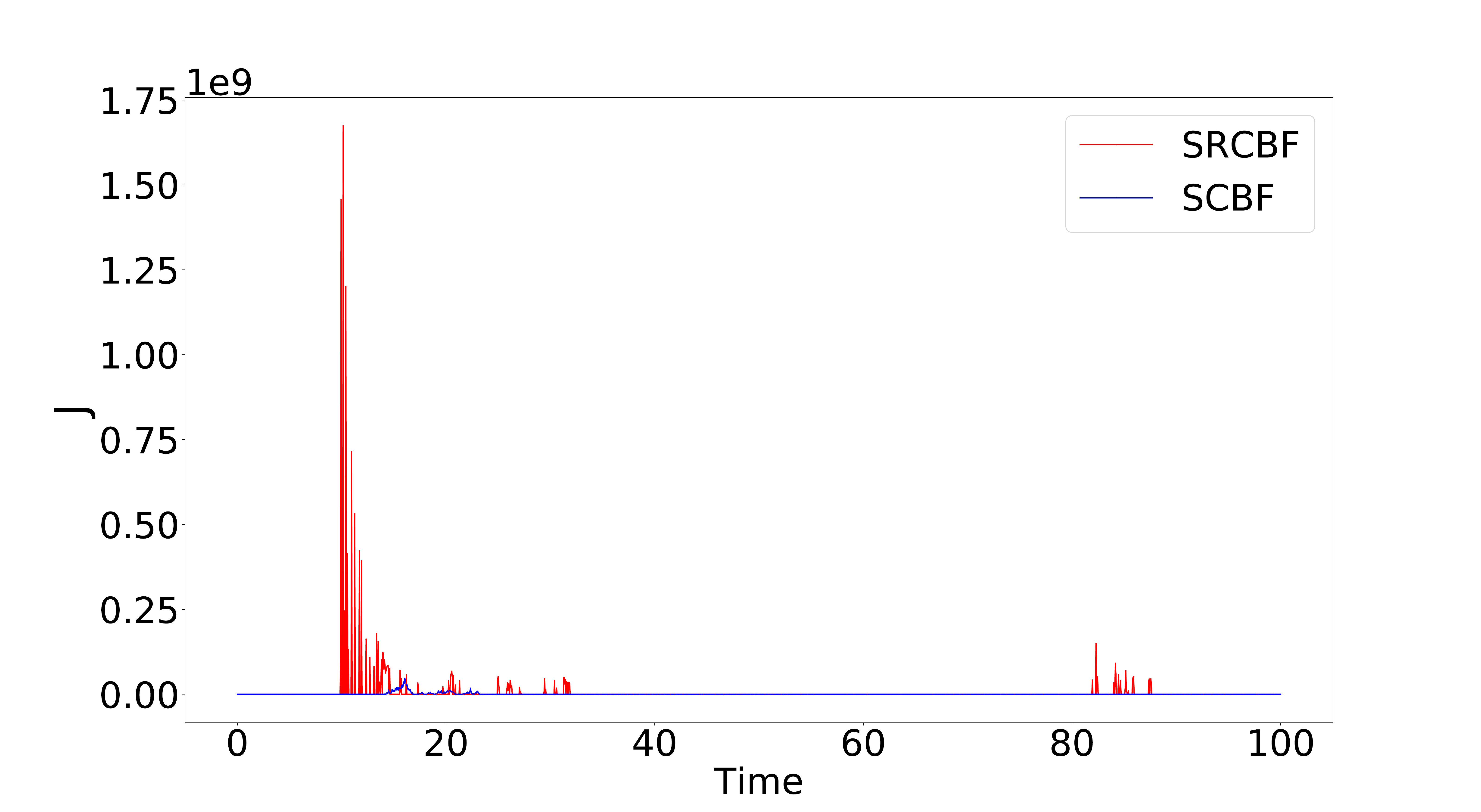}
	\caption{Plot of control effort $J=u^2$ for Example 1.}
    \label{fig:J}

\end{figure}

\begin{table}[htbp]
\centering
\begin{tabular}{lcccl}
\toprule
& SRCBF & SCBF \\ 
\midrule
Unbounded control & 90\% & 70\% &\\ 
Bounded control & 25\% & 65\% &\\
\bottomrule
\end{tabular}
\caption{Safe probability of Example 1 between SRCBF and SCBF under bounded (saturated) and unbounded control inputs. We sample 20 trajectories for each case and calculate the safe probability. The simulation step time is chosen to be $t=0.0005s$. We can see from this table that the proposed SCBF is more robust to control input saturation than SRCBF, which tend to require unbounded control inputs at the boundary of the safe set. See also Section \ref{sec:ex} for a simple illustrative example.}
\label{tab}
\end{table}

\subsection{Example 2}
In the second example, we test our SCBF using a differential drive model as in \cite{lavalle2006planning}:
\begin{equation*}
\begin{split}
ds=d\begin{bmatrix}
x\\
y\\
\theta\\
\end{bmatrix}=\begin{bmatrix}
\cos{\theta} & 0\\
\sin{\theta} & 0\\
0 & 1\\
\end{bmatrix}\begin{bmatrix}
v\\
w\\
\end{bmatrix}dt+\begin{bmatrix}
    \sigma_1&0&0\\
    0&\sigma_2&0\\
    0&0&0\\
    \end{bmatrix}dW
\end{split},
\end{equation*}
where $x$ and $y$ are the planar positions of the center of the vehicle, $\theta$ is its orientation, $v=2$ is its forward velocity, the angular velocity $w$ is the control of the system and $W$ is a standard Brownian motion representing uncertainty in $x$ and $y$. The working space of the vehicle is a circle centered at $(0,0)$ with a radius of $r=3$. Our objective is to control the vehicle within the working space and as a result, the safety requirement can be encoded using $h(s)=r^2-x^2-y^2$. We solve QP problems using constraints that are obtained from Theorem~\ref{thm} as
\begin{equation}
    \mathcal{A}\cdot\mathcal{A}(9-x^2-y^2)\geq 0
\end{equation}
We first test safe probability for different initial positions. We use our SCBF to test this probability under different noise levels within $[0,0.3]$. For simplicity, we set $\sigma_1=\sigma_2=\sigma$. For each value of $\sigma$, we sample 1000 trajectories and calculate the safe probability. We also compare the result between SCBF and a high-order SZCBF, which is an extension of the work \cite{santoyo2021barrier} as 
\begin{equation*}
    \begin{split}
        h_1(s)&=\mathcal{A}h(s)+\alpha_1h(s),\\
        h_2(s)&=\mathcal{A}h_1(s)+\alpha_2h(s).
    \end{split}
\end{equation*}
The constraints from SZCBF that $h_2(s)\geq 0$ are used to solve QP problems. We first compare the safe probability between SCBF and SZCBF for different noise level within $[0,0.2]$. For each value of $\sigma$, we randomly sample 1000 initial points and generate 1000 trajectories accordingly. We calculate safe probability over this 1000 trajectories and plot the result as in Figure~\ref{fig:compare}. We can find out that SCBF has a better safe probability over SZCBF. In another experiment, we randomly sample 10 initial points. Then we generate 500 trajectories using SCBF and SZCBF for each initial point. We fix the noise to be $\sigma=0.2$. The result is shown as in Figure~\ref{fig:init}. From the figure, we can find out that SCBF has a overall better performance than SZCBF under randomly sampled initial points. 
% We use two initial positions $(1.5,1.5,-\frac{\pi}{2})$ and $(0,2,-\frac{\pi}{6})$ to test the safe probability using SCBF. In both cases, the probability is more than 60\% when noise increase as much as 0.3. The result is shown in Figure~\ref{fig:prob}. We can also estimate the lower bound of safe probability based on Theorem\ref{thm}. For initial point s(0)=$(1.5,1.5,-\frac{\pi}{2})$, we have $b_0(s_0)=4.5$. The maximum value of $b_0(s)$ is $c_0=9$. Also $b_1(s_0)=1.5$ and the maximum value of $b_1(s)=-2v(x\cos{\theta}+y\sin{\theta})$ is $c_1=12$. As a result, the estimation of lower bound of safe probability is $\frac{4.5}{9}\cdot\frac{1.5}{12}=6.25\%$.

% We also compare the safe probability between SCBF and SZCBF using different initial points. We test noise level within $[0,0.2]$. For each value of $\sigma$, we sample 1000 initial points and calculate the safe probability. From Figure~\ref{fig:compare}, we can find out that when $\sigma=0.2$, the safe probability drops around 10\% using SZCBF while the SCBF has a safe probability higher than 50\%.

% \begin{figure}[htbp]
% 	\centering
% 	\includegraphics[width=1\linewidth, scale=0.5]{example_car_ratio.pdf}
% 	\caption{Safe probability using SCBF for different initial points under different noise levels. We test 2 different initial points  and plot safe probability under noise levels from 0 to 0.3.}
% 	\label{fig:prob}
% \end{figure}%

\begin{figure}[htbp]
    \centering
	\includegraphics[width=1\linewidth, scale=0.5]{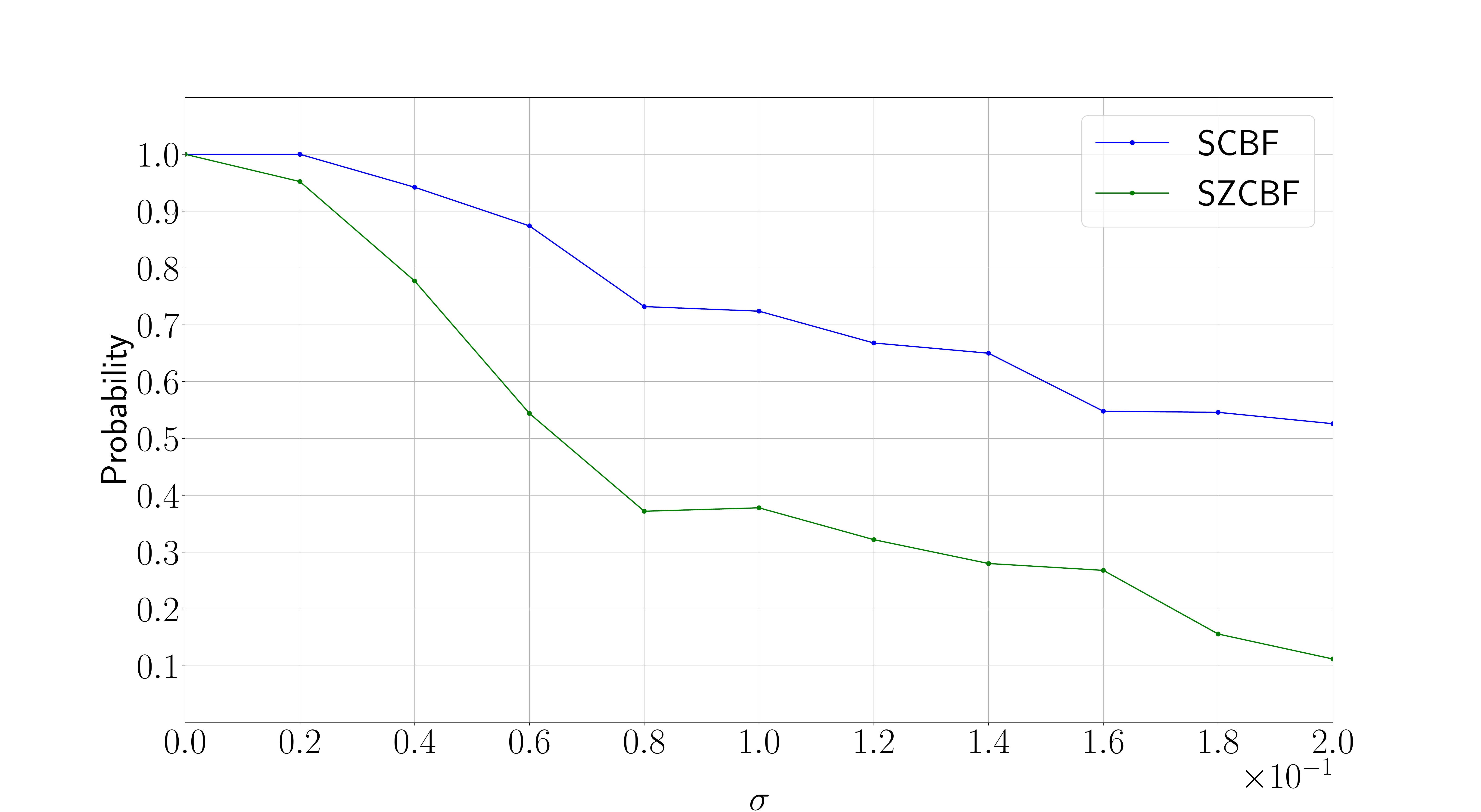}
	\caption{Safe probability between SCBF and SZCBF. We compare noise level within [0,0.2]. For each value of $\sigma$, we sample 1000 initial points to calculate safe probability.}
    \label{fig:compare}

\end{figure}

\begin{figure}[htbp]
	\centering
	\includegraphics[width=1\linewidth, scale=0.5]{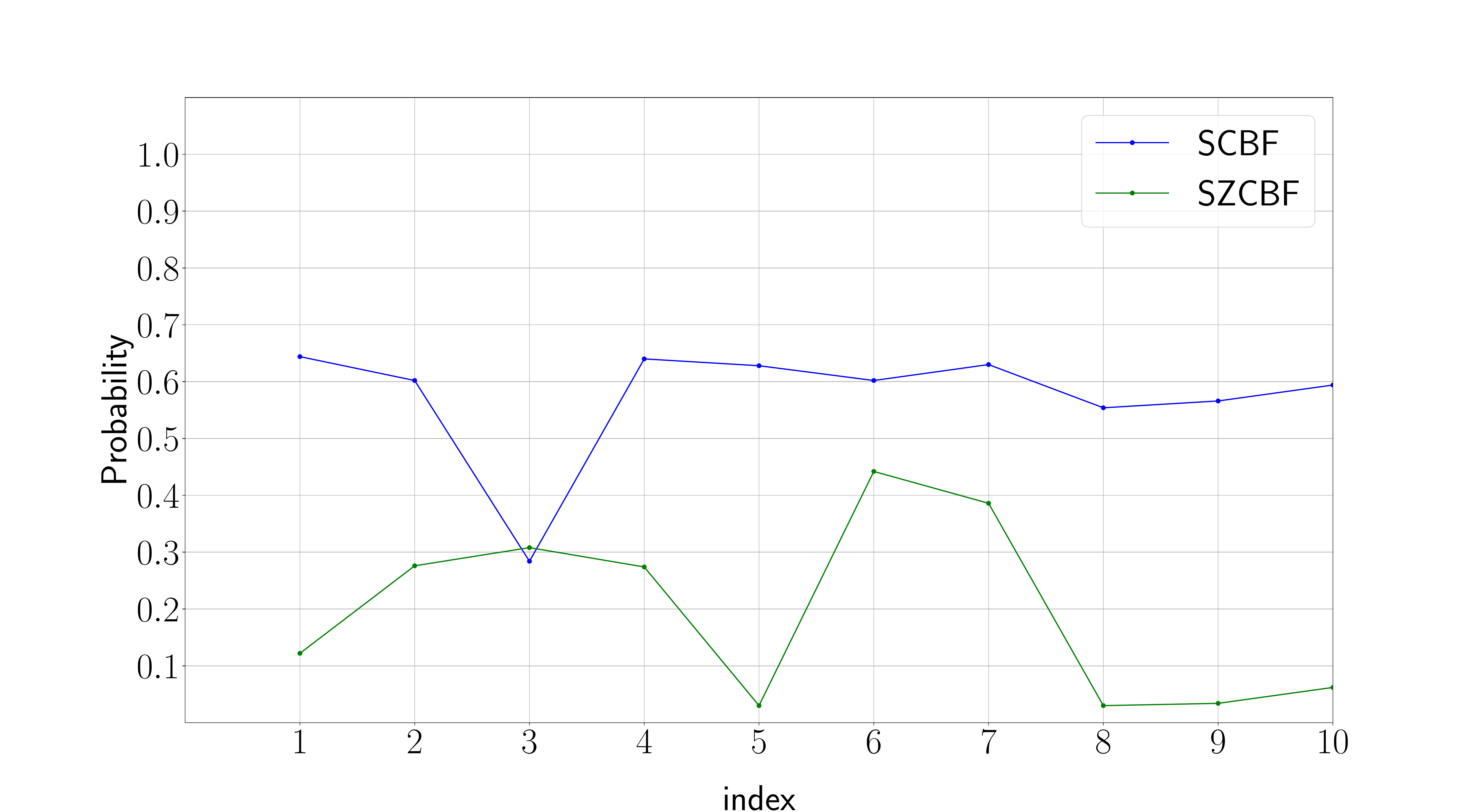}
	\caption{Safe probability of 10 randomly sample initial points. For each initial point, we sample 500 trajectories using SCBF and SZCBF respectively. The horizontal axis represents the index of the initial points. }
	\label{fig:init}
\end{figure}

\section{Conclusion}
In this work, considering the pros and cons of the existing formulations for stochastic barrier functions (such as frequently used SRCBFs and SZCBFs), we propose stochastic control barrier functions (SCBFs) for safety-critical control of stochastic systems and extend the worst-case safety probability estimation to high-order SCBFs. We show that the proposed SCBFs provide good trade-offs between the imposed control constraints and the conservatism in the estimation of safety probability, which are demonstrated both theoretically and empirically. In particular, the proposed scheme is utilized to control an automatic cruise control model and a differential drive mobile robot model. %The controller properties as well as the safety probability results from simulations match with the theorem.  

While the estimate of the safety probability in the high relative degree case appears to be conservative, as partly discussed in Remark \ref{rem: thm} and \ref{rem: controller}, worst-case scenarios do exist when the control scheme generated by SCBFs loses dominance on the intermediate supper-level sets. To accurately obtain the probabilistic winning sets, it is necessary to capture how the probability measure is distorted by the input processes. However, this may be computationally challenging for stochastic control systems with high relative degree. 

For future work, it is intriguing to connect Lyapunov-type characterizations with exit-time problems, such that despite of a direct bottom-level control, the probability of satisfaction for the controlled paths (dependent on time and initial positions) on each intermediate super level set can be provided with more accuracy. %The idea has been applied for stochastic reach-avoid specifications using optimal control scheme to find the probability as a value function (of stochastic HJB equation) \cite{esfahani2016stochastic,esfahani2010robust}. 
To embed the Lyapunov scenario into the existing HJB solver, topological analysis is needed to reduce the effect of discontinuity\cite{esfahani2016stochastic}. In addition, since the existing exit-time provides an off-line synthesis of controller, a direct combination of Lyapunov-type controller synthesis via quadratic programming and off-line exit-time probability solution may not be desirable. It would be applicable to consider learning algorithm to maintain the ignorable error whilst alleviate computational complexity. At last, more complex stochastic specification and control synthesis via Lyapunov could be investigated as in the deterministic cases.

%\newpage
\bibliography{root}{}

\begin{thebibliography}{10}

\bibitem{alpern1985defining}
Bowen Alpern and Fred~B Schneider.
\newblock Defining liveness.
\newblock {\em Information Processing Letters}, 21(4):181--185, 1985.

\bibitem{ames2019control}
Aaron~D Ames, Samuel Coogan, Magnus Egerstedt, Gennaro Notomista, Koushil
  Sreenath, and Paulo Tabuada.
\newblock Control barrier functions: Theory and applications.
\newblock In {\em Proc. of ECC}, pages 3420--3431. IEEE, 2019.

\bibitem{ames2014control}
Aaron~D Ames, Jessy~W Grizzle, and Paulo Tabuada.
\newblock Control barrier function based quadratic programs with application to
  adaptive cruise control.
\newblock In {\em Proc. of CDC}, pages 6271--6278. IEEE, 2014.

\bibitem{ames2016control}
Aaron~D Ames, Xiangru Xu, Jessy~W Grizzle, and Paulo Tabuada.
\newblock Control barrier function based quadratic programs for safety critical
  systems.
\newblock {\em IEEE Transactions on Automatic Control}, 62(8):3861--3876, 2016.

\bibitem{chen2017obstacle}
Yuxiao Chen, Huei Peng, and Jessy Grizzle.
\newblock Obstacle avoidance for low-speed autonomous vehicles with barrier
  function.
\newblock {\em IEEE Transactions on Control Systems Technology},
  26(1):194--206, 2017.

\bibitem{clark2019control}
Andrew Clark.
\newblock Control barrier functions for complete and incomplete information
  stochastic systems.
\newblock In {\em Proc. of ACC}, pages 2928--2935. IEEE, 2019.

\bibitem{esfahani2016stochastic}
Peyman~Mohajerin Esfahani, Debasish Chatterjee, and John Lygeros.
\newblock The stochastic reach-avoid problem and set characterization for
  diffusions.
\newblock {\em Automatica}, 70:43--56, 2016.

\bibitem{garcia2015comprehensive}
Javier Garc{\i}a and Fernando Fern{\'a}ndez.
\newblock A comprehensive survey on safe reinforcement learning.
\newblock {\em Journal of Machine Learning Research}, 16(1):1437--1480, 2015.

\bibitem{girard2006efficient}
Antoine Girard, Colas Le~Guernic, and Oded Maler.
\newblock Efficient computation of reachable sets of linear time-invariant
  systems with inputs.
\newblock In {\em Proc. of HSCC}, pages 257--271. Springer, 2006.

\bibitem{hsu2015control}
Shao-Chen Hsu, Xiangru Xu, and Aaron~D Ames.
\newblock Control barrier function based quadratic programs with application to
  bipedal robotic walking.
\newblock In {\em Proc. of ACC}, pages 4542--4548. IEEE, 2015.

\bibitem{kushner1967stochastic}
Harold~J Kushner.
\newblock {\em Stochastic Stability and Control}.
\newblock Academic Press, 1967.

\bibitem{lamport1977proving}
Leslie Lamport.
\newblock Proving the correctness of multiprocess programs.
\newblock {\em IEEE Transactions on Software Engineering}, (2):125--143, 1977.

\bibitem{lavalle2006planning}
Steven~M LaValle.
\newblock {\em Planning Algorithms}.
\newblock Cambridge University Press, 2006.

\bibitem{nguyen2016exponential}
Quan Nguyen and Koushil Sreenath.
\newblock Exponential control barrier functions for enforcing high
  relative-degree safety-critical constraints.
\newblock In {\em Proc. of ACC}, pages 322--328. IEEE, 2016.

\bibitem{prajna2007framework}
Stephen Prajna, Ali Jadbabaie, and George~J Pappas.
\newblock A framework for worst-case and stochastic safety verification using
  barrier certificates.
\newblock {\em IEEE Transactions on Automatic Control}, 52(8):1415--1428, 2007.

\bibitem{ratschan2007safety}
Stefan Ratschan and Zhikun She.
\newblock Safety verification of hybrid systems by constraint propagation-based
  abstraction refinement.
\newblock {\em ACM Transactions on Embedded Computing Systems (TECS)},
  6(1):8--es, 2007.

\bibitem{rauscher2016constrained}
Manuel Rauscher, Melanie Kimmel, and Sandra Hirche.
\newblock Constrained robot control using control barrier functions.
\newblock In {\em Proc. of IROS}, pages 279--285. IEEE, 2016.

\bibitem{santoyo2021barrier}
Cesar Santoyo, Maxence Dutreix, and Samuel Coogan.
\newblock A barrier function approach to finite-time stochastic system
  verification and control.
\newblock {\em Automatica}, 125:109439, 2021.

\bibitem{sarkar2020high}
Meenakshi Sarkar, Debasish Ghose, and Evangelos~A Theodorou.
\newblock High-relative degree stochastic control {L}yapunov and barrier
  functions.
\newblock {\em arXiv preprint arXiv:2004.03856}, 2020.

\bibitem{taylor2020learning}
Andrew Taylor, Andrew Singletary, Yisong Yue, and Aaron Ames.
\newblock Learning for safety-critical control with control barrier functions.
\newblock In {\em Proc. of L4DC}, pages 708--717. PMLR, 2020.

\bibitem{wang2020learning}
Chuanzheng Wang, Yinan Li, Yiming Meng, Stephen~L Smith, and Jun Liu.
\newblock Learning control barrier functions with high relative degree for
  safety-critical control.
\newblock In {\em Proc. of ECC, to appear}, 2021.

\bibitem{xiao2019control}
Wei Xiao and Calin Belta.
\newblock Control barrier functions for systems with high relative degree.
\newblock In {\em Proc. of CDC}, pages 474--479. IEEE, 2019.

\end{thebibliography}
\bibliographystyle{plain}
\end{document}